\documentclass[journal,twocolumn]{IEEEtran}
\IEEEoverridecommandlockouts
\usepackage{amsfonts}
\usepackage[tbtags]{amsmath}
\usepackage{graphicx,subfigure}
\usepackage{color}      
\usepackage{epsfig}
\usepackage{psfig, psfrag}
\usepackage{amssymb}
\usepackage{bbm}
\usepackage{algorithm2e}
\usepackage{epstopdf}

\newcommand{\hG}{\hat{G}}
\newcommand{\hV}{\hat{V}}
\newcommand{\hE}{\hat{E}}

\newcommand{\bfbeta}{\boldsymbol{\beta}}
\newcommand{\sumxi}{\sum_{i=1}^3 X_i}

\def\stlv{$st$-lab($v$)~}
\def\stl{$st$-lab}

\def\endproof{\hspace*{\fill}~$\blacksquare$}
\long\def\comment#1{}

\newcommand{\beq}{\begin{equation}}
\newcommand{\eeq}{\end{equation}}
\newcommand{\beqno}{\begin{equation*}}
\newcommand{\eeqno}{\end{equation*}}

\newcommand{\bes}{\begin{split}}
\newcommand{\ees}{\end{split}}
\newcommand{\bdm}{\begin{displaymath}}
\newcommand{\edm}{\end{displaymath}}

\newcommand{\goes}{\rightarrow}

\newtheorem{theorem}{Theorem}
\newtheorem{claim}{Claim}
\newtheorem{lemma}{Lemma}

\newtheorem{definition}{Definition}
\newtheorem{remark}{Remark}

\newcommand{\bd}{\begin{definition}}
\newcommand{\ed}{\end{definition}}

\newcommand{\bv}{\begin{vugraph}}
\newcommand{\ev}{\end{vugraph}}
\newcommand{\bi}{\begin{itemize}}
\newcommand{\ei}{\end{itemize}}
\newcommand{\ben}{\begin{enumerate}}
\newcommand{\een}{\end{enumerate}}

\newcommand{\bean}{\begin{eqnarray*} }
\newcommand{\eean}{\end{eqnarray*} }
\newcommand{\bea}{\begin{eqnarray} }
\newcommand{\eea}{\end{eqnarray} }
\newcommand{\beali}{\begin{align} }
\newcommand{\eeali}{\end{align} }

\newcommand{\ba}{\begin{array} }
\newcommand{\ea}{\end{array} }

\begin{document}

\title{Communicating the sum of sources over a network}
\author{\authorblockN{Aditya Ramamoorthy, {\it {Member, IEEE}}}, and \authorblockN{Michael Langberg, {\it {Member, IEEE}}} \thanks{A. Ramamoorthy is with the Department of Electrical and Computer Engineering, Iowa State University, Ames IA 50011, USA (email: adityar@iastate.edu). \newline M. Langberg is with the Computer Science Division, Open University of Israel, Raanana 43107, Israel (email: mikel@openu.ac.il).

The material in this work was presented in part at the 2008 IEEE International Symposium on Information Theory in Toronto, Canada, at the 2009 IEEE International Symposium on Information Theory in Seoul, South Korea and the 2010 IEEE International Symposium on Information Theory, Austin, TX, USA. This work was supported in part by NSF grants CCF-1018148 and DMS-1120597.}}
\maketitle
\thispagestyle{empty}
\begin{abstract}
We consider the network communication scenario, over directed acyclic networks with unit capacity edges in which a number of sources $s_i$ each holding independent unit-entropy information $X_i$ wish to communicate the sum $\sum{X_i}$ to a set of terminals $t_j$. We show that in the case in which there are only two sources or only two terminals, communication is possible if and only if each source terminal pair $s_i/t_j$ is connected by at least a single path. For the more general communication problem in which there are three sources and three terminals, we prove that a single path connecting the source terminal pairs does not suffice to communicate $\sum{X_i}$.
We then present an efficient encoding scheme which enables the communication of $\sum{X_i}$ for the three sources, three terminals case, given that each source terminal pair is connected by {\em two} edge disjoint paths. 
\end{abstract}
{\keywords network coding, function computation, multicast, distributed source coding.}
\section{Introduction}

We consider the problem of function computation over directed acylic networks in this work.
Under our setting the sources are independent and the network links are error-free, but capacity constrained. However, the topology of the network can be quite complicated, e.g., an arbitrary directed acyclic graph. This serves as an abstraction of current-day computer networks at the higher layers. We investigate the problem of characterizing the {\it network resources} required to communicate the sum (over a finite field) of a certain number of sources over a network to multiple terminals. By network resources, we mean the number of edge disjoint paths between various source terminal pairs in the network. Our work can be considered as using network coding to compute and multicast sums of the messages, as against multicasting the messages themselves.

The problem of multicast has been studied intensively under the paradigm of
network coding. The seminal work of Ahlswede et al. \cite{al}
showed that under network coding the multicast capacity is the
minimum of the maximum flows from the source to each individual
terminal node. The work of Li et al. \cite{lc} showed that linear
network codes are sufficient to achieve the multicast capacity.
The algebraic approach to network coding proposed by Koetter and
M\'{e}dard \cite{rm} provided simpler proofs of these results.

The problem of multicasting sums of sources is an important component in enabling the multicast of correlated sources over a network (using network coding). Network coding for correlated sources was first
examined by Ho et al. \cite{hoMKKESL06}. The work of Ramamoorthy
et al. \cite{adisepDSC} showed that in general separating
distributed source coding and network coding is suboptimal except
in the case of two sources and two terminals. The work of Wu et al. \cite{WuSXK09} presented a practical approach
to multicasting correlated sources over a network. Reference \cite{WuSXK09} also stated
the problem of communicating sums over networks using network coding, and called it the {\it Network Arithmetic} problem.
We elaborate on related work in the upcoming Section~\ref{sec:back}.

In this work, we present (sometimes tight) upper and lower bounds on the network resources required for communicating the sum of sources over a network under certain special cases.
\subsection{Main Contributions}
We consider networks that can be modeled as directed acyclic graphs, with unit capacity edges. Let $G = (V, E)$ represent such a graph. There is a set of source nodes $S \subset V$ that observe independent unit-entropy sources, $X_i, i \in S$, and a set of terminal nodes $T\subset V$, that seek to obtain $\sum_{i \in S} X_i$, where the sum is over a finite field.
Our work makes the following contributions.
\begin{itemize}
\item[i)] {\it Characterization of necessary and sufficient conditions when either $|S| = 2$ or $|T| = 2$.}\\ Suppose that $G$ is such that there are either two sources ($|S| = 2$) and an arbitrary number of terminals or an arbitrary number of sources and two terminals ($|T| = 2$). The following conditions are necessary and sufficient for recovery of $\sum_{i \in S} X_i$ at all terminals in $T$.
    \beqno \text{max-flow}
(s_i - t_j) \geq 1 \text{~for all~} s_i \in S \text{~and~} t_j \in T. \eeqno
Our proofs are constructive, i.e., we provide efficient algorithms for the network code assignment.
\item[ii)] {\it Unit connectivity does not suffice when $|S|$ and $|T|$ are both greater than 2}.\\
We present a network $G$ such that $|S|=|T|=3$ in which the maximum flow between each source terminal pair is at least 1 and (as opposed to that stated above) communicating the sum of sources is not possible.
\item[iii)] {\it Sufficient conditions when $|S| = |T| = 3$.}\\ Suppose that $G$ is such that $|S| = |T| = 3$. The following condition is sufficient for recovery of $\sum_{i \in S} X_i$ at all $t_j \in T$.
    \beqno \text{max-flow}
(s_i - t_j) \geq 2 \text{~for all~} s_i \in S \text{~and~} t_j \in T. \eeqno
Efficient algorithms for network code assignment are presented in this case as well. Note however, that the algorithms may be randomized in some cases, with a probability of success that can be made arbitrarily close to one.
\end{itemize}


This paper is organized as follows. We discuss background and related work in Section \ref{sec:back} and our network coding model in Section \ref{sec:net-coding-model}. The characterization for the case of $|S| = 2, |T| = n$ and $|S| = n, |T| = 2$ is discussed in Section \ref{sec:case_2s_or_2t}. Our counter-example demonstrating that unit-connectivity does not suffice for three sources and three terminals can be found in Section \ref{sec:counter}. Sections \ref{sec:main} and \ref{sec:case_3} discuss the sufficient characterization in the case of three sources and three terminals, and Section \ref{sec:conclusion} presents the conclusions and possibilities for future work.
\section{Background and Related Work}
\label{sec:back}


Prior work of an information theoretic flavor in the area of function computation has mainly considered the case of two correlated sources $X$ and $Y$, with direct links between the sources and the terminal, where the terminal is interested in reconstructing a function $f(X, Y)$. In these works, the topology of the network is very simple, however the structure of the correlation between $X$ and $Y$ may be arbitrary.
In this setting, Korner \& Marton \cite{kornerM79} determine the rate region for encoding the modulo-2 sum of $X$ and $Y$ when they are uniform, correlated binary sources.
The work of Orlitsky \& Roche \cite{orlitskyR01} determines the required rate for sending $X$ to a decoder with side information $Y$ that must reliably compute $f(X,Y)$. The result of \cite{orlitskyR01} was extended to the case when both $X$ and $Y$ need to be encoded (under certain conditions) in \cite{doshiSMJ07}. Yamamoto \cite{yamamoto82} (generalizing the Wyner-Ziv result \cite{wynerZ76}) found the rate-distortion function for sending $X$ to a decoder with side information $Y$, that wants to compute $f(X,Y)$ within a certain distortion level (see also \cite{fengES04} for an extension). Nazer et al. \cite{nazerG07} consider the problem of reliably reconstructing a function over a multiple-access channel (MAC) and finding the capacity of finite-field multiple access networks. In the majority of these works, the sources and the terminal are connected by direct links or by simple networks (such as a MAC). A work closer in spirit to our work is \cite{feiziM09} that considers functional compression over tree-networks.


In this work we consider a problem setting in which the sources are independent and the network links are error-free, but capacity constrained. However, the topology of the network can be quite complicated, such as an arbitrary directed acyclic graph. This is well motivated since it is a good abstraction of current-day computer networks (at the higher layers). We investigate the problem of characterizing the {\it network resources} required to communicate the sum of a certain number of sources over a network to multiple terminals.
Network resources can be measured in various ways. For example, one may specify the maximum flow between the subsets of the source nodes and  subsets of the terminal nodes in the network. In the current work, all of our characterizations are in terms of the maximum flow between various $s_i - t_j$ pairs, where $s_i$ ($t_j$) denotes a source (terminal) node.
Previous work in this area, includes the work of Ahlswede et al. \cite{al}, who introduced the concept of network coding and showed the capacity region for multicast. In multicast, the terminals are interested in reconstructing the actual sources. Numerous follow-up works have extended and improved the results of \cite{al}, in different ways. For example, \cite{lc, rm} considered multicast with linear codes. Ho et al. \cite{hoMKKESL06} proposed random network coding and examined the multicast of correlated sources over a network and showed a tight capacity region for it that can be achieved by using random network codes. Follow-up works \cite{adisepDSC,WuSXK09} investigated practical approaches for the multicast of correlated sources. As shown in \cite{WuSXK09}, the problem of communicating (multicasting) the sum (over a finite field) of sources over a network is a subproblem that can help facilitate practical approaches to the problem of multicasting correlated sources.

In this work we consider function computation under network coding.
Specifically, we present network code assignment algorithms for the problem of multicasting the sum of sources over a network. As one would expect, one needs fewer resources in order to support this. To the best of our knowledge, the first work to examine function computation in this setting is the work of Ramamoorthy \cite{ramamoorthy08}, that considered the problem of multicasting sums of sources, when there are either two sources or two terminals in the network. Subsequently, the work of Langberg and Ramamoorthy \cite{langbergR09} showed that the characterization of \cite{ramamoorthy08} does not hold in the case of three sources and three terminals. Reference \cite{langbergR09}, proposed an alternate characterization in this case. The current paper is a revised and extended version of \cite{ramamoorthy08}, \cite{langbergR09} and \cite{langbergR10} that contains all the proofs and additional observations.

We note, as presented by Rai and Dey in \cite{raiD12}, that the task of finding a network coding scheme in the setting of sum-networks is strongly connected to the problem of finding a network coding solution in the multiple-unicast communication setting. Specifically, for any mutiple unicast network, \cite{raiD12} constructs a sum-network which is solvable if and only if the original multiple unicast network is solvable (the reduction of  \cite{raiD12} increases the number of sources and terminals in the network). Rai and Dey \cite{raiD09} independently found the same counter-example found in our work \cite{langbergR09}; however, their proof only shows that linear codes do not suffice for multicasting sums under the characterization of \cite{ramamoorthy08}. 
The work of Appuswamy et al. \cite{appuswamyFKZ08, appuswamyFKZ09} also considers the problem of computing general functions in the setting of error-free directed acyclic networks. In \cite{appuswamyFKZ08, appuswamyFKZ09}, the emphasis is on considering the rate of the computation, where the rate refers to the maximum number of times a function can be computed per network usage. While their setting is significantly more general, their results are mostly in the context of only single terminal networks.

Finally, the work most related to our result on three source/three
terminal networks is the conference publication of Shenvi and Dey
\cite{SD10} (and its extended version available as \cite{SD10b}) which
proposes (in this case) a combinatorial characterization for sum
computation via network coding. In our work, for three source/three
terminal networks, we present a simple sufficient combinatorial
condition for sum-communication based on flow requirements.  Our result
is not proven to be necessary, and indeed in the subsequent work of
\cite{SD10,SD10b}, our flow condition is refined (and weakened) to
obtain a tight characterization.
The characterization of \cite{SD10,SD10b} implies a significant
improvement in the understanding of 3s/3t sum-networks.
Nevertheless, we believe that our results (obtained independently and
prior to \cite{SD10,SD10b}) are of interest due to the natural and
simple nature of our sufficient condition.

\section{Network coding model}
\label{sec:net-coding-model} Our model and terminology follow those common in the network coding literature, e.g. \cite{rm}.
We represent the network as a directed acyclic graph G = (V,E).
The network contains a set of source nodes $S \subset V$ that are observing independent, discrete
unit-entropy sources and a set of terminals $T \subset V$.
We assume that each edge in the network has unit capacity and can
transmit one symbol from a finite field of size $q$ per unit
time. We are free to choose $q$ large enough. In addition, as we shall see in the later discussion, in some cases we may need to choose $q$ to be an odd prime. If a given edge has
a higher capacity, it can be treated as multiple unit capacity
edges.
A directed edge $e$ between nodes $v_i$ and $v_j$ is
represented as $(v_i \goes v_j)$. Thus $head(e) = v_j$ and
$tail(e) = v_i$. A path between two nodes $v_i$ and $v_j$ is a
sequence of edges $\{ e_1, e_2, \dots, e_k\}$ such that $tail(e_1)
= v_i, head(e_k) = v_j$ and $head(e_i) = tail(e_{i+1}), i = 1,
\dots,  k-1$.

Our counter-example in Section \ref{sec:counter} considers arbitrary network codes. However, our constructive algorithms in Sections \ref{sec:case_2s_or_2t} and \ref{sec:main} shall use linear network codes.
In linear network coding, the signal on an edge $(v_i \goes v_j)$, is a linear combination
of the signals on the incoming edges on $v_i$ and the source
signal at $v_i$ (if $v_i \in S$). In this paper we assume that the
source (terminal) nodes do not have any incoming (outgoing) edges from (to) other nodes. If this is not the
case one can always introduce an artificial source (terminal) connected to
the original source (terminal) node by an edge of sufficiently large capacity that has no incoming (outgoing) edges. We shall only
be concerned with networks that are directed acyclic in which internal nodes have sufficient memory. Such networks can be treated as delay-free networks. Let $Y_{e_i}$ (such
that $tail(e_i) = v_k$ and $head(e_i) = v_l$) denote the signal on
the $i^{th}$ edge in $E$ and let $X_j$ denote the $j^{th}$ source.
Then, we have
\begin{align*}
Y_{e_i} &= \sum_{\{e_j | head(e_j) = v_k\}} f_{j,i} Y_{e_j} \text{~if $v_k \in V \backslash S$}, \text{~and}\\
Y_{e_i} &= \sum_{\{j | X_j \text{~observed at~} v_k\}} a_{j,i} X_j
\text{~ if $v_k \in S$},
\end{align*}
where the coefficients $a_{j,i}$ and $f_{j,i}$ are from $GF(q)$.
Note that since the graph is directed acyclic, it is possible to
express $Y_{e_i}$ for an edge $e_i$ in terms of the sources
$X_j$'s. Suppose that there are $n$ sources $X_1, \dots, X_n$. If
$Y_{e_i} = \sum_{k=1}^n \beta_{e_i, k} X_k$ then we say that the
global coding vector of edge $e_i$ is $\boldsymbol{\beta}_{e_i} =
[ \beta_{e_i, 1} ~\cdots~ \beta_{e_i, n}]$. For brevity we shall mostly use the term coding vector instead of global coding
vector in this paper. We say that a node $v_i$ (or edge $e_i$) is
downstream of another node $v_j$ (or edge $e_j$) if there exists a
path from $v_j$ (or $e_j$) to $v_i$ (or $e_i$).



\section{Networks with either two sources/$n$ terminals or $n$ sources/two terminals}
\label{sec:case_2s_or_2t}

In this section we state and prove the result for (a) networks with two sources and $n$ terminals, and (b) networks with $n$ sources and two terminals. Before embarking on this proof, we overview the concept of greedy encoding that will be used throughout the paper when considering two source networks.
\begin{definition}
\label{def:greedy}
{\it Greedy encoding.}
Consider a graph $G = (V, E)$, with two source nodes $s_1$ and $s_2$ and an edge $e' = (u \rightarrow v) \in E$. Suppose that the coding vector on each edge $e$ entering $u$, has only $0$ or $1$ entries, i.e., $\bfbeta_e = [\beta_{e,1} ~ \beta_{e,2}]$, where $\beta_{e,i} \in \{0,1\}$, for all $i=1,2$. We say that the encoding on edge $e'$ is greedy, if for $i = 1, 2$ we have
\begin{equation}
\beta_{e',i} = \begin{cases}
0 & \text{~if~} \beta_{e,i} = 0, \forall e \text{~entering~} u\\
1 & \text{~otherwise.}
\end{cases}
\end{equation}
A coding vector assignment for $G$, is said to be greedy if the encoding on each edge in $G$ is greedy.
\end{definition}

Consider a vertex $u$ that is downstream of a subset of the source nodes, $B \subseteq \{1, 2\}$. Under greedy coding it can be seen that the outgoing edges of $u$ will carry the sum $\sum_{i \in B}{X_i}$. Namely, if a node only receives either $X_1$ or $X_2$, it just forwards them. Alternatively, if it receives both of them or $X_1 + X_2$, then it just transmits $X_1 + X_2$.

\noindent The first result of this section is the following.
\begin{theorem}
\label{thm:2s_n_t} Consider a directed acylic graph $G = (V, E)$
with unit capacity edges, two source nodes $s_1$ and $s_2$ and $n$
terminal nodes $t_1, \dots , t_n$ such that \beqno \text{max-flow}
(s_i - t_j) \geq 1 \text{~for all~} i=1,2 \text{~and~} j=1, \dots,
n. \eeqno Assume that at each source node $s_i$, there is a unit-rate source
$X_i$, and that the $X_i$'s are independent.
Then, there exists an assignment of coding vectors to all edges
such that each $t_j, j = 1, \dots, n$ can recover $X_1 + X_2$.
\end{theorem}


\noindent {\it Proof of Theorem \ref{thm:2s_n_t}}.
Consider any terminal node $t_j$.
As we assume that $\text{max-flow} (s_i - t_j) \geq 1 \text{~for all~} i=1,2$, it holds that $t_j$ is downstream of both $s_1$ and $s_2$.
Thus, (using greedy encoding) by the observation above, $t_j$ can recover $X_1+X_2$.
\endproof

Note that if any of the conditions in the statement of Theorem \ref{thm:2s_n_t} are
violated then some terminal will be unable to
compute $X_1+X_2$. For example, if max-flow$(s_1 - t_j) < 1$
then any decoded signal $Y$ at $t_j$ will have $H(Y|X_2)<1$ (as $Y$ is solely a function of $X_1$ and $X_2$).
We conclude that $Y$ cannot be $X_1+X_2$.

Next, consider the class of networks with $n$ sources and two terminals. The original proof of this result (obtained in \cite{ramamoorthy08}) was obtained via a series of graph-theoretic operations on the network. However, subsequently it was shown in \cite{raiD09} that this result follows in a simpler manner by using the idea of network reversibility. We state the result below.

\begin{theorem}
\label{thm:n_s_2_t} Consider a directed acylic graph $G = (V, E)$
with unit capacity edges, $n$ source nodes $s_1, s_2,
\dots, s_n$ and two terminal nodes $t_1$ and $t_2$ such that
\beqno \text{max-flow} (s_i - t_j) \geq 1 \text{~for all~} i = 1,
\dots, n \text{~and~}j= 1,2. \eeqno
Assume that the source nodes observe independent unit-entropy sources $X_i, i = 1, \dots, n$.
Then, there exists an
assignment of coding vectors such that each terminal can recover
the
sum of the sources $\sum_{i=1}^n X_i$.
\end{theorem}
{\it Proof.} Given a directed acyclic network $G = (V, E)$, its reverse network $\tilde{G}$ is defined as the network that has the same set of vertices $V$, but the orientation of each edge is reversed. Moreover the sources in $G$ become terminals in $\tilde{G}$ and the terminals in $G$ become the sources in $\tilde{G}$. Reference \cite{raiD09} shows if the sum of sources in $G$ can be multicast to all the terminals (in $G$), the sum of sources in $\tilde{G}$ can also be multicast to  all the terminals (in $\tilde{G}$). Our proof now follows from using reversibility and Theorem \ref{thm:2s_n_t}.
\endproof 

%

\section{Insufficiency of unit-connectivity for 3-source/3-terminal networks}
\label{sec:counter}
In the discussion below we show an instance of a network with three sources and three terminals, with at least one path connecting each source terminal pair, in which the sum of sources cannot (under any network code) be transmitted (with zero error) to all three terminals.
Consider the network shown in Figure \ref{fig:3s-3t-counter-eg}, with three source nodes and three terminal nodes such that the source nodes observe unit entropy sources $X_1, X_2$ and $X_3$ that are also independent. All edges are unit capacity. 
As showed in Figure \ref{fig:3s-3t-counter-eg} the incoming edges into terminal $t_3$ contain the values $f(X_1, X_2)$ and $f'(X_2, X_3)$ where $f$ and $f'$ are some functions of the sources.

Suppose that $X_3 = 0$. This implies that $t_1$ should be able to recover $X_1 + X_2$ (that has entropy 1) from just $f(X_1, X_2)$. Moreover note that each edge is unit capacity. Therefore, the entropy of $f(X_1, X_2)$ also has to be 1, i.e., there exists a one-to-one mapping between the set of values that $f(X_1, X_2)$ takes and the values of $X_1 + X_2$. In a similar manner we can conclude that there exists a one-to-one mapping between the set of values that $f'(X_2, X_3)$ takes and the values of $X_2 + X_3$. At terminal $t_3$, there needs to exist some function $h(f(X_1, X_2), f'(X_2, X_3)) = \sumxi$. By the previous observations, this also implies the existence of a function $h'(X_1 + X_2, X_2 + X_3)$ that equals $\sumxi$. However, this is a contradiction.
Consider the following sets of inputs: $X_{1} = a, X_2 = 0, X_3 = c$ and $X_1' = a - b, X_2' = b, X_3' = c - b$. In both cases the inputs to the function $h'(\cdot, \cdot)$ are the same. However $\sum_{i=1}^3 X_i = a + c$, while $\sum_{i=1}^3 X_i' = a - b + c$, that are in general different. Therefore such a function $h'(\cdot, \cdot)$ cannot exist.

Note that we have presented the proof in the context of scalar nonlinear network codes. However, even if we consider vector sources along with vector network codes, the same idea of the proof can be used.
\begin{figure}[t]
\centering
\includegraphics[width=60mm, clip=true]{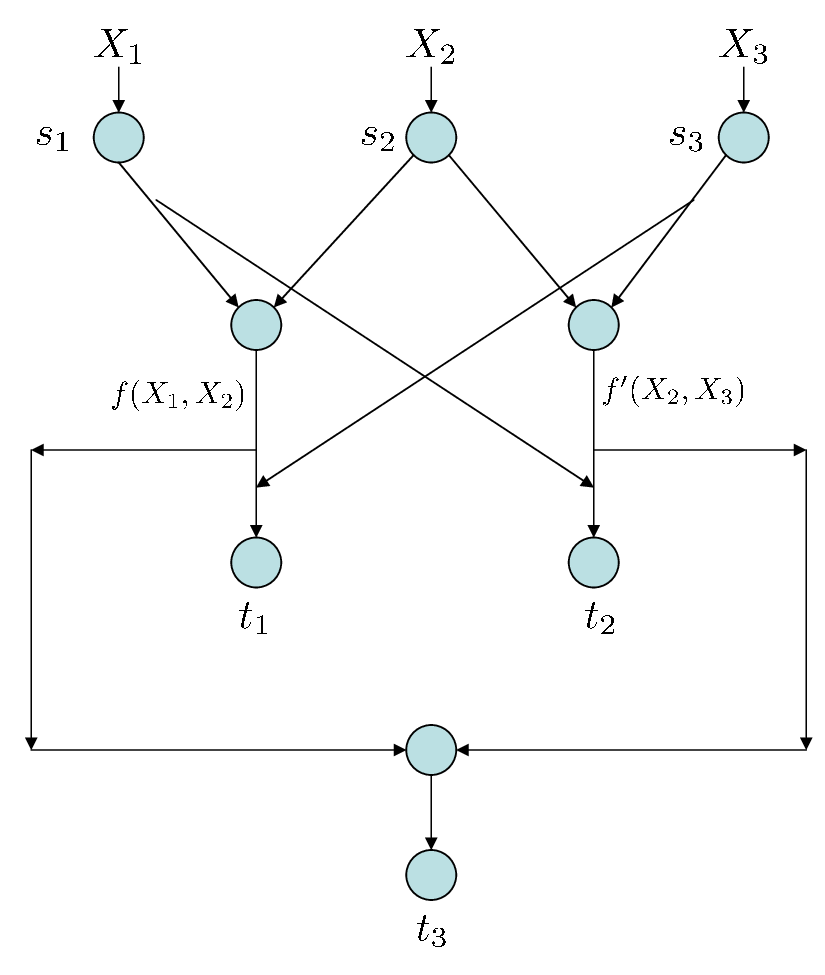} \vspace{-3mm}
\caption{\label{fig:3s-3t-counter-eg}Example of a network with three sources and three terminals, such that there exists at least one path between each source and each terminal. However all the terminals cannot compute $\sumxi$.}
\end{figure}


\section{Case of three sources and three terminals}
\label{sec:main}
It is evident from the counter-example discussed in Section \ref{sec:counter}, that the characterization of the required resources for networks with three sources and three terminals is different from the cases discussed in Section \ref{sec:case_2s_or_2t}. In this section, we show that as long as each source is connected by {\em two} edge disjoint paths to each terminal, the terminals can recover the sum. We present efficient linear encoding schemes, i.e., linear codes that can be found in time polynomial in the number of nodes, that allow communication in this case. The main result of this section can be summarized as follows.
\begin{theorem}
\label{the:main}
Let $G=(V,E)$ be a directed acyclic network with three sources $s_1,s_2,s_3$ and three terminals $t_1,t_2,t_3$.
Let $X_i$ be the (unit entropy) information present at source $s_i$.
If there exist two edge disjoint paths between each source/terminal pair, then there exists a network coding scheme in which the sum $X_1+X_2+X_3$ is obtained at each terminal $t_j$.
Moreover, such a network code can be found efficiently.
\end{theorem}
\begin{remark}
Our example in Section \ref{sec:counter}, shows that a single path between each $s_i - t_j$ pair does not suffice. At the other extreme, if there are three edge-disjoint paths between each $s_i - t_j$ pair, then one can actually multicast $X_1, X_2$ and $X_3$ to each terminal \cite{rm}. Our results show that {\it two} edge disjoint paths between each source terminal pair are sufficient for multicasting sums. 
\end{remark}

We start by giving an overview of our proof. Our approach for determining the desired network code has three steps.
In the first step, we turn our graph $G$ into a graph $\hG=(\hV,\hE)$ in which each internal node $v \in \hV$ is of total degree (in-degree + out-degree) at most three.
We refer to such graphs as {\em structured} graphs.
Our efficient reduction follows that appearing in \cite{LSB06}, and has the following properties:
(a) $\hG$ is acyclic.
(b) For every source (terminal) in $G$ there is a corresponding source (terminal) in $\hG$.
(c) For any two edge disjoint paths $P_1$ and $P_2$ connecting a source terminal pair in $G$, there exist two {\em vertex} disjoint paths in $\hG$ connecting the corresponding source terminal pair. Here and throughout we say two paths between a source terminal pair are vertex disjoint even though they share their first and last vertices (i.e., the source and terminal at hand).
(d) Any feasible network coding solution in $\hG$ can be efficiently turned into a feasible network coding solution in $G$. We note that the same reduction has facilitated a study of three-source, three-terminal multiple unicast networks \cite{huangR11, huangR12}.

It is not hard to verify that proving Theorem~\ref{the:main} on structured graphs implies a proof for general graphs $G$ as well.
Indeed, given a network $G$ satisfying the requirements of Theorem~\ref{the:main} construct the corresponding network $\hG$.
By the properties above, $\hG$ also satisfies the requirements of Theorem~\ref{the:main}.
Assuming Theorem~\ref{the:main} is proven for structured graphs $\hG$, we conclude the existence of a feasible network code in $\hG$.
Finally, this network code can be converted (by property (d) above) into a feasible network code for $G$ as desired.
The mapping between $G$ and $\hG$ is presented in detail in \cite{LSB06}.
For notational reasons, from this point on in the discussion we will assume that our input graph $G$ is structured --- which is now clear to be w.l.o.g.

In the second step of our proof, we give edges and vertices in the graph $G$ certain labels depending on the combinatorial structure of $G$.
This step can be viewed as a decomposition of the graph $G$ (both the vertex set and the edge set) into certain {\em class} sets 
that will play a major role in our analysis.
The decomposition of $G$ is given in detail in Section~\ref{sec:dec}.

Finally, in the third and final step of our proof, using the labeling above we perform a case analysis for the proof of Theorem~\ref{the:main}.
Namely, based on the terminology set in Section~\ref{sec:dec}, we identify several scenarios, and prove Theorem~\ref{the:main} assuming they hold. As the different scenarios we consider will cover all possible ones, we will conclude our proof. Our detailed case analysis is given in Section~\ref{sec:case} and Section~\ref{sec:case_3}.


\subsection{Graph decomposition}
\label{sec:dec}

As justified in our previous discussions, we assume throughout that 
any internal vertex in $V$ (namely, any vertex which is neither a source or a sink) has total degree at most $3$.
Moreover, we assume $G$ satisfies the connectivity requirements specified in Theorem~\ref{the:main}.

We start by labeling the vertices of $G$.
A vertex $v \in V$ is labeled by a pair $(c_s,c_t)$ specifying how many sources (terminals) it is {\em connected} to.
Specifically, $c_s(v)$ equals the number of sources $s_i$ for which there exists a path connecting $s_i$ and $v$ in $G$.
Similarly, $c_t(v)$ equals the number of terminals $t_j$ for which there exists a path connecting $v$ and $t_j$ in $G$.
For example, any source is labeled by the pair $(1,3)$, and any terminal by the pair $(3,1)$.
An internal vertex $v$ labeled $(\cdot,1)$ is connected to a single terminal only. This implies that any information leaving $v$ will reach at most a single terminal.

\subsection{Case analysis}
\label{sec:case}
Our proof methodology involves a classification of networks based on the node labeling procedure presented above.
For each class of networks we shall argue that each terminal can compute the sum of the sources $(X_1 + X_2 + X_3)$. Our proof shall be constructive, i.e., it can be interpreted as an algorithm for finding the network code that allows each terminal to recover $(X_1 + X_2 + X_3)$.
\subsubsection{Case 0}
There exists a node of type $(3, 3)$ in $G$.
Suppose node $v$ is of type $(3, 3)$. This implies that there exist $path(s_i - v)$, for $i = 1, \dots, 3$ and $path(v - t_j)$, for  $j = 1, \dots, 3$. Consider the subgraph induced by these paths and color each edge on $\cup_{i = 1}^3 path(s_i - v)$ red and each edge on $\cup_{j = 1}^3 path(v - t_j)$ blue. We claim that as $G$ is acyclic, at the end of this procedure each edge gets only one color. To see this suppose that a red edge is also colored blue.
This implies that it lies on a path from a source to $v$ and a path from $v$ to a terminal, i.e. its existence implies a directed cycle in the graph.

Now, we can find an inverted tree that is a subset of the red edges directed into $v$ and similarly a tree rooted at $v$ with $t_1, t_2$ and $t_3$ as leaves using the blue edges. Finally, we can compute $(X_1 + X_2 + X_3)$ at $v$ over the red tree and multicast it to $t_1, t_2$ and $t_3$ over the blue subgraph. More specifically, one may use an encoding scheme in which internal nodes of the red tree receiving $Y_1$ and $Y_2$ send on their outgoing edge the sum $Y_1+Y_2$.

\subsubsection{Case 1}
There exists a node of type $(2, 3)$ in $G$.
Note that it is sufficient to consider the case when there does not exist a node of type $(3,3)$ in $G$. We shall show that this case is equivalent to a two sources, three terminals problem.

W.l.o.g. we suppose that there exists a $(2, 3)$ node $v$ that is connected to $s_2$ and $s_3$. We color the edges on $path(s_2 - v)$ and $path(s_3 - v)$ blue. Next, consider the set of paths $\cup_{i=1}^3 path(s_1 - t_i)$. We claim that these paths do not have any intersection with the blue subgraph. This is because the existence of such an intersection would imply that there exists a path between $s_1$ and $v$ which in turn implies that $v$ would be a $(3, 3)$ node.
We can now compute $(X_2 + X_3)$ at $v$ by finding a tree consisting of blue edges that are directed into $v$. Suppose that the blue edges are removed from $G$ to obtain a graph $G'$. Since $G$ is directed acyclic, we have that there still exists a path from $v$ to each terminal after the removal. Now, note that (a) $G'$ is a graph such that there exists at least one path from $s_1$ to each terminal and at least one path from $v$ to each terminal, and (b) $v$ can be considered as a source that contains $(X_2 + X_3)$. 
Now, $G'$ satisfies the condition given in Theorem~\ref{thm:2s_n_t} (which addresses the two sources version of the problem at hand), therefore we are done.
\subsubsection{Case 2}
There exists a node of type $(3, 2)$ in $G$.
As before it suffices to consider the case when there do not exist any $(3, 3)$ or $(2, 3)$ nodes in the graph. Suppose that there exists a (3,2) node $v$ and w.l.o.g. assume that it is connected to $t_1$ and $t_2$. We consider the subgraph $G'$ induced by the union of the following sets of paths
\begin{enumerate}
\item $\cup_{i=1}^3 path(s_i - v)$, \item $\cup_{i=1}^2 path(v - t_i)$, and \item $\cup_{i=1}^3 path(s_i - t_3)$.
\end{enumerate}
Note that as argued previously, a subset of edges of $\cup_{i=1}^3 path(s_i - v)$ can be found so that they form a tree directed into $v$. For the purposes of this proof, we will assume that this has already been done, i.e., the graph $\cup_{i=1}^3 path(s_i - v)$ is a tree directed into $v$.

The basic idea of the proof is to show that the paths from the sources to terminal $t_3$, i.e., $\cup_{i=1}^3 path(s_i - t_3)$ are such that their overlap with the other paths is very limited. Thus, the entire graph can be decomposed into two parts, one over which the sum is transmitted to $t_1$ and $t_2$ and another over which the sum is transmitted to $t_3$. 

Towards this end, note that $path(s_1 - t_3)$ cannot have an intersection with either $path(s_2 - v)$ or $path(s_3 - v)$, for if such an intersection occurred at a node $v'$, then $v'$ would be a node of type $(2,3)$ contradicting our assumption. Likewise, it can be noted that (a) $path(s_2 - t_3)$ cannot have an intersection with either $path(s_1 - v)$ or $path(s_3 - v)$, and (b) $path(s_3 - t_3)$ cannot have an intersection with either $path(s_1 - v)$ or $path(s_2 - v)$. In a similar manner, we observe that the paths $path(s_1 - t_3), path(s_2 - t_3)$ and $path(s_3 - t_3)$ cannot have an intersection with either $path(v - t_1)$ or $path(v - t_2)$ as this would imply that $v$ is a $(3,3)$ node contradicting our assumption.

%
%
%
%

We now discuss the coding solution on $G'$. Let $v_i$ be the node closest to $v$ that belongs to both $path(s_i-v)$ and $path(s_i-t_3)$ (notice that $v_i$ may equal $s_i$ but it cannot equal $v$).
On the paths $path(s_i-v_i)$ send $X_i$. On the paths $path(v_i-v)$ send information that will allow $v$ to obtain $X_1+X_2+X_3$. This can be easily done, as these (latter) paths form a tree into $v$. Namely, one may use an encoding scheme in which internal nodes receiving $Y_1$ and $Y_2$ send on their outgoing edge the sum $Y_1+Y_2$. By the discussion above (and the fact that $G'$ is acyclic) it holds that the information flowing on edges $e$ in $path(v_i-t_3), i = 1, \dots, 3$ has not been specified by the encoding defined above. Thus, one may send information on the paths $path(v_i-t_3)$ that will allow $t_3$ to obtain $X_1+X_2+X_3$. Here we assume the paths $path(v_i-t_3)$ form a tree into $t_3$, if this is not the case we may find a subset of edges in these paths with this property. Once more, by the discussion above (and the fact that $G'$ is acyclic) it holds that the information flowing on edges $e$  in the paths $path(v-t_1)$ and $path(v-t_2)$ has not been specified (by the encodings above). On these edges we may transmit the sum $X_1+X_2+X_3$ present at $v$.



\subsubsection{Case 3}
There do not exist $(3,3), (2, 3)$ and $(3,2)$ nodes in $G$.
Note that thus far we have not utilized the fact that there exist two edge-disjoint paths from each source to each terminal in $G$. In previous cases, the problem structure that has emerged due to the node labeling, allowed us to communicate $(X_1 + X_2 + X_3)$ by using just one path between each $s_i - t_j$ pair. However, for the case at hand we will indeed need to use the fact that there exist two paths between each $s_i - t_j$ pair.  As we will see, this significantly complicates the analysis, and we present it in the upcoming section.

The following definitions are required for this case.
An edge $e=(u,v)$ for which $v$ is labeled $(\cdot,1)$ will be referred to as a {\em terminal} edge.
Namely, any information flowing on $e$ can reach at most a single terminal.
If this terminal is $t_j$ then we will say that $e$ is a $t_j$-edge.
Clearly, the set of $t_1$-edges is disjoint from the set of $t_2$-edges (and similarly for any pair of terminals).
An edge which is not a terminal edge will be referred to as a {\em remaining} edge or an $r$-edge for short.


Note that there exists an ordering of edges in $E$ in which any $r$-edge comes before any terminal edge, and in addition there is no path from a terminal edge to an $r$-edge. This is obtained by an appropriate topological order in $G$.
Moreover, for any terminal $t_j$, the set of $t_j$-edges form a connected subgraph of $G$ with $t_j$ as its sink.
To see this note that by definition each $t_j$-edge $e$ is connected to $t_j$ and all the edges on a path between $e$ and $t_j$ are $t_j$-edges.
Finally, the head of an $r$-edge is either of type $(\cdot, 2)$ or $(\cdot, 3)$ (as otherwise it would be a terminal edge).

For each terminal $t_j$ we define a set of vertices referred to as the leaf set $L_j$ of $t_j$. 
\begin{definition}  {\it Leaf set of a terminal.}
The leaf set of terminal $t_j$ is the set of nodes of in-degree 0 in the subgraph consisting of $t_j$-edges.
\end{definition}
We note that a source node can be a leaf node for a given terminal.

\section{Analysis of Case 3}
\label{sec:case_3}
Note that the node labeling procedure presented above assigns a label $(c_s(v), c_t(v))$ to a node $v$ where $c_s(v)$ ($c_t(v)$) is the number of sources (terminals) that $v$ is connected to. This labeling ignores the actual identity of the sources and terminals that have connections to $v$. It turns out that we need to use an additional, somewhat finer notion of node connectivity when we want to analyze case 3. We emphasize that throughout this section, we still operate under the assumption the graph is structured (cf. reduction discussed in Section \ref{sec:main}).


Towards this end, for case 3 (i.e., in a graph $G$ without $(3,3), (2,3)$ and $(3,2)$ nodes) we introduce the notion of the source-terminal label (or $st$-label for short) of a node.
For each $(2,2)$ node in $G$, the $st$-label of the node is defined as the $4$-tuple of sources and terminals it is connected to, e.g., if $v$ is connected to sources $s_1$ and $s_2$ and terminals $t_1$ and $t_2$, then its $st$-label, denoted \stlv is $(s_1, s_2, t_1, t_2)$. We shall also say that the source label of $v$ is $(s_1, s_2)$ and the terminal label of $v$ is $(t_1, t_2)$. The following claim is immediate.



\begin{claim}
\label{claim:leaf_existence}
If there is a $(2,2)$ node $v$ in $G$ of $st$-label, \stlv, then each terminal in the terminal label of $v$ has at least one leaf with $st$-label \stlv. For example, if \stlv = $(s_1, s_2, t_1, t_2)$, then both $t_1$ and $t_2$ have leaves with $st$-label $(s_1, s_2, t_1, t_2)$.
\end{claim}
\begin{proof}
W.l.o.g, let \stlv = $(s_1, s_2, t_1, t_2)$. This implies that there exists a path $P$ between $v$ and $t_1$.
Let $\ell$ be a leaf of $t_1$ on $P$.
It follows directly from the definition of a leaf that $\ell$ is the last node on $P$ with terminal label at least $2$, namely $c_t(\ell)\geq 2$. Namely, if $c_t(\ell) =1$ then the incoming link of $\ell$ on $P$ would be a $t_1$-edge (contradicting the assumption that $\ell$ is a leaf).
Moreover, $c_t(\ell)$ is exactly $2$ and no larger as otherwise $c_t(v)$ would also be greater than 2 contradicting our assumptions in the claim.
This implies that the terminal label of $\ell$ is exactly $(t_1,t_2)$.
As $\ell$ is downstream of $v$ it holds that $c_s(\ell) \geq c_s(v) = 2$.
Here also, it holds that $c_s(\ell)$ is exactly $2$, otherwise $\ell$ would be a $(3,2)$ node (contradicting our assumption for case 3).
This implies that the source label of $\ell$ is $(s_1,s_2)$.
Therefore, $t_1$ has a leaf of label $(s_1, s_2, t_1, t_2)$. A similar argument holds for $t_2$.
\end{proof}

The notion of an $st$-label is useful for the set of graphs under case 3, since we can show that there can never be an edge between nodes of different $st$-labels. 
\begin{claim}
\label{claim:sep_color_subgraphs}
Consider a graph $G$, with sources, $s_i, i = 1, \dots, 3$, and terminals $t_j, j = 1, \dots 3$, such that it does not have any $(3,3), (2,3)$ or $(3,2)$ nodes. There does not exist an edge between $(2,2)$ nodes of different $st$-labels in $G$.
\end{claim}
\begin{proof}
Assume otherwise and consider two $(2,2)$ nodes $v_1$ and $v_2$ such that \stl($v_1$) $\neq$ \stl($v_2$), for which there is an edge $(v_1, v_2)$ in $G$. Note that if the source labels of \stl($v_1$) and \stl($v_2$) are different, then $v_2$ has to be a $(3,2)$ node, which is a contradiction. Likewise, if the terminal labels of \stl($v_1$) and \stl($v_2$) are different, then $v_1$ has to be a $(2,3)$ node, which is also a contradiction.
\end{proof}

Claim~\ref{claim:sep_color_subgraphs} implies that we are free to assign any coding coefficients on a subgraph induced by nodes of one $st$-label, without having to worry about the effect of this on another subgraph induced by nodes of a different $st$-labels (simply because there is no such effect).

Our approach is as follows.
We divide the set of graphs under case 3, into various classes, depending on the number of distinct $st$-labels that exist in the graph. It turns out that as long as the number of $st$-labels in the graph is not 2, i.e., either 0,1 or 3 and higher,
then there is a simple argument which shows that each terminal can be satisfied. The argument in the case of two distinct $st$-labels is a bit more involved and is developed separately. It can be shown that our counter-example in Section \ref{sec:counter} is a case where there are two $st$-labels. Note however, that in our counter-example there are certain $s_i - t_j$ pairs that have only one path between them (and thus the sum $X_1+X_2+X_3$ cannot be computed at all terminals). 
\begin{claim}
\label{claim:color_multicast}
Consider the subgraph induced by the vertices with a certain $st$-label, w.l.o.g. $(s_1, s_2, t_1, t_2)$ in $G$, denoted by $G_{(s_1, s_2, t_1, t_2)}$. There exists an assignment of encoding vectors over $G_{(s_1, s_2, t_1, t_2)}$, such that any (unit entropy) function of the sources $X_1$ and $X_2$ can be multicasted to all nodes in $G_{(s_1, s_2, t_1, t_2)}$. Moreover, such encoding vector assignments can be done independently over subgraphs of different $st$-labels.
\end{claim}
\begin{proof}
Note that we are working with directed acyclic graphs. Thus, there is a node $v^*$ in $G_{(s_1, s_2, t_1, t_2)}$, such that it has no incoming edges in $G_{(s_1, s_2, t_1, t_2)}$. Next, note that the path from $s_1$ to $v^*$ has no intersection with a path from $s_2$ or $s_3$. To see this, suppose that there was such an intersection at node $v'$. If there is a path from $s_3$ to $v'$, then $v^*$ is a $(3,2)$ node (which contradicts the assumption that $v^*$ is a $(2,2)$ node).
If there is a path from $s_2$ to $v'$, then $v'$ and the remaining vertices connecting $v'$ to $v^*$ on the path from $s_1$ to $v^*$ have $st$-label $(s_1, s_2, t_1, t_2)$.
Contradicting the fact that $v^*$ has no incoming edges in $G_{(s_1, s_2, t_1, t_2)}$. Likewise, we see that the path from $s_2$ to $v^*$ has no intersection with a path from $s_1$ or $s_3$.

Therefore, the path from $s_1$ to $v^*$ carries $X_1$ exclusively,
and likewise for the path from $s_2$ to $v^*$.
Thus, $v^*$ can obtain both $X_1$ and $X_2$ and can compute any (unit entropy) function of them.
Moreover, $v^*$ can transmit this function to all nodes of $G_{(s_1, s_2, t_1, t_2)}$ downstream of $v^*$.
As the argument above can be repeated for any node $v^*$ of in-degree 0 in $G_{(s_1, s_2, t_1, t_2)}$ it follows that all nodes of $G_{(s_1, s_2, t_1, t_2)}$ can obtain the desired function of $X_1$ and $X_2$.


Finally, we note that the encoding functions assigned to edges in subgraphs of different $st$-labels can be done independently,
since there does not exist any edge between nodes of different $st$-labels (from Claim \ref{claim:sep_color_subgraphs}), and all $(1,\cdot)$ edges use the same encoding scheme regardless of the $st$-label at hand.
\end{proof}

\begin{lemma}
\label{lemma:one_or_three_colors}
Consider a graph $G$, with sources, $s_i, i = 1, \dots, 3$, and terminals $t_j, j = 1, \dots 3$, such that (a) it does not have any $(3,3), (2,3)$ or $(3,2)$ nodes, and (b) there exists at least one $s_i - t_j$ path for all $i$ and $j$. Consider the set of all $(2,2)$ nodes in $G$ and their corresponding $st$-labels. If there exist no $st$-labels, exactly one $st$-label or at least three distinct $st$-labels in $G$, then there exists a set of coding vectors such that each terminal can recover $\sumxi$.
\end{lemma}
\begin{proof}
Note that all leaves in $G$ are of type $(1,2), (1,3)$ or $(2,2)$.
This implies that any terminal $t_j$ that does not have a $(2,2)$ leaf with source $st$-label including $s_i$, must have a $(1,\cdot)$ leaf (i.e., a leaf connected to a single source) at which $X_i$ can be recovered, for instance by simply forwarding the source information along the path to the leaf.
We refer to such leaves as {\em singleton} $X_i$ leaves.
The above follows directly by the connectivity assumption (b) stated in the Lemma. Recall that in Section \ref{sec:net-coding-model}, we presented the network coding model as one where each symbol flowing on an edge is from a field of size $q$. In cases 2 and 3 in the analysis below, we assume that the characteristic of the field of operation is $> 2$. This can for instance be done by choosing $q=3$.
\begin{figure}[t]
\centering
\includegraphics[width=80mm, clip=true]{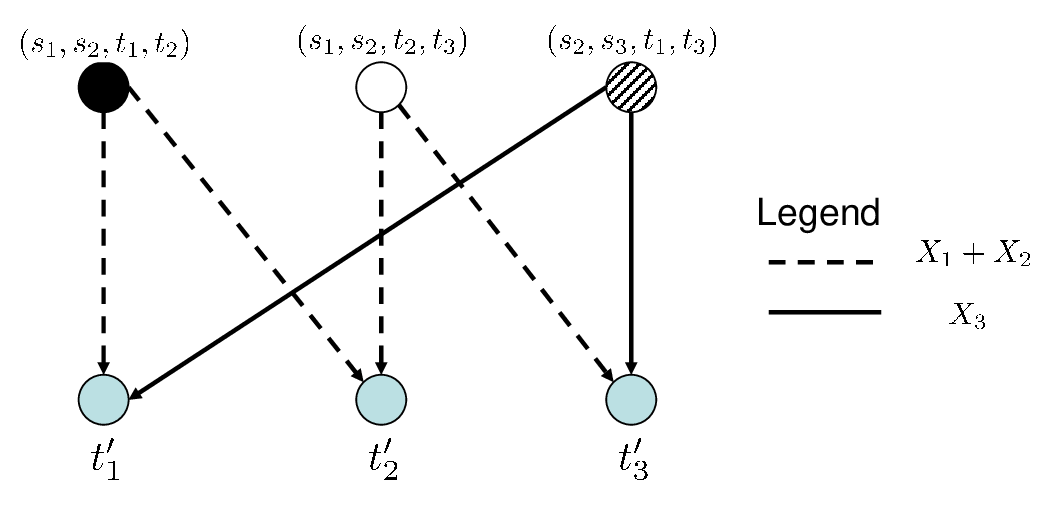} \vspace{-3mm}
\caption{\label{fig:aux_bipartite_graph_1} A possible instance of $G_{aux}$ when the degree sequence of the terminals is $(2, 2, 2)$. The encoding specified in the legend denotes the encoding to be used on the appropriate subgraphs.}
\end{figure}
\begin{itemize}
\item[(0)] {\it Case 0.} There are no $st$-labels in $G$.\\
This implies that there are no $(2,2)$ nodes in $G$ and thus all terminals $t_j$ have distinct leaves holding $X_1$, $X_2$, and $X_3$ respectively. It suffices to design a simple code on the paths from those leaves to $t_j$ which enables $t_j$ to recover the sum $X_1+X_2+X_3$.
\item[(i)] {\it Case 1.} There is only one $st$-label in $G$.\\
In this case perform greedy encoding (cf. Definition \ref{def:greedy}) on the $r$-edges. We show that each terminal can recover $\sumxi$ from the content of its leaves.
W.l.o.g, suppose that the $st$-label is $(s_1, s_2, t_1, t_2)$. Using Claim \ref{claim:leaf_existence}, this means that both $t_1$ and $t_2$ have leaves of this $st$-label. The greedy encoding implies that $t_1$ and $t_2$ can obtain $X_1 + X_2$ from the corresponding leaves. Moreover, both $t_1$ and $t_2$ have a singleton leaf containing $X_3$, because of the connectivity requirements. Therefore, they can compute $\sumxi$. The terminal $t_3$ has only singleton leaves, such that there exists at least one $X_1, X_2$ and $X_3$ leaf. Thus it can compute their sum.
\item[(ii)] {\it Case 2.} There exist exactly three distinct $st$-labels in $G$.\\
It is useful to introduce an auxiliary bipartite graph that denotes the existence of the $st$-labels at the leaves of the different terminals. This bipartite graph denoted $G_{aux}$ is constructed as follows. There are three nodes $t_i', i = 1, \dots, 3$ that denote the terminals on one side and three nodes $c_i', i = 1, \dots, 3$ that denote the $st$-labels on the other side. If the $st$-label $c_i'$ has $t_j$ in its support, then there is an edge between $c_i'$ and $t_j'$, i.e., $t_j$ has a leaf of $st$-label $c_i'$. See Figure~\ref{fig:aux_bipartite_graph_1}.
The following properties of $G_{aux}$ are immediate.
\begin{figure}[t]
\centering
\includegraphics[width=80mm, clip=true]{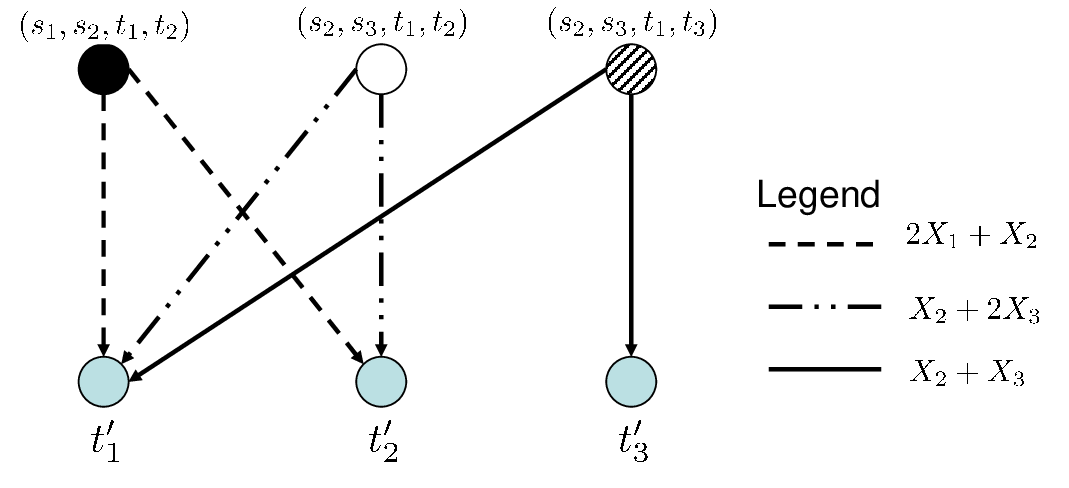} \vspace{-3mm}
\caption{\label{fig:aux_bipartite_graph_2}A possible instance of $G_{aux}$ when the degree sequence of the terminals is $(3, 2, 1)$. The encoding specified in the legend denotes the encoding to be used on the appropriate subgraphs.}
\end{figure}
\begin{itemize}
\item Each $c_i'$ has degree-2.
\item Each $t_i'$ has degree at most 3 (as there are 3 distinct $st$-labels).
\item Multiple edges between nodes are disallowed.
\end{itemize}
Note that there are exactly three possible source $st$-labels ($(s_1, s_2), (s_2, s_3)$ and $(s_3, s_1)$) and three possible terminal $st$-labels ($(t_1, t_2), (t_2, t_3)$ and $(t_3, t_1)$).
We now perform a case analysis depending upon the degree sequence of nodes $t_j', j = 1, \dots, 3$ in $G_{aux}$. The degree sequence is specified by a 3-tuple, where we note that the sum of the entries has to be 6.

\begin{itemize}
\item[a)] The degree sequence is a permutation of $(0,3,3)$.\\
This only happens if the terminal label of all $st$-labels, $c_i', i = 1, \dots, 3$ is the same and in turn implies that the source label of each $st$-label is distinct, i.e., the source $st$-labels include $(s_1, s_2), (s_2, s_3)$ and $(s_1, s_3)$. In this case, greedy encoding (cf. Definition \ref{def:greedy}) works for the two terminals in the $st$-label support. This is because each terminal will obtain $X_1+X_2$, $X_2+X_3$ and $X_1+X_3$ at its leaves 
(using Claims \ref{claim:leaf_existence} and \ref{claim:color_multicast}) from which the terminal can compute $2\sumxi$. The remaining terminal is not connected to any (2,2) leaf, which implies that all its leaves contain singleton values, from which it can compute $\sumxi$.
\item[b)] The degree sequence is $(2,2,2)$.\\
This only happens if all the terminal labels of the $st$-labels are distinct, i.e., the terminal labels are $(t_1, t_2)$, $(t_2, t_3)$ and $(t_1, t_3)$. Now consider the possibilities for the source labels.

\indent If there is only one source label, then greedy encoding ensures that the sum of exactly two of the sources reaches each terminal. The connectivity condition guarantees that the remaining source is available as a singleton at a leaf of each terminal. Therefore we are done.

\indent If there are exactly two distinct source $st$-labels, then we argue as follows (see Figure~\ref{fig:aux_bipartite_graph_1}). On the subgraphs induced by the $st$-labels with the same source label, perform greedy encoding. On the remaining subgraph, propagate the remaining useful source. We illustrate this with an example that is w.l.o.g. Suppose that the $st$-labels are $(s_1, s_2, t_1, t_2)$, $(s_1, s_2, t_2, t_3)$ and $(s_2, s_3, t_1, t_3)$. We perform greedy encoding on the subgraphs of the first two $st$-labels, and only propagate $X_3$ on the subgraph of the third $st$-label. As shown in Figure \ref{fig:aux_bipartite_graph_1}, this means that terminals $t_1$ and $t_3$ are satisfied. Note that the connectivity condition dictates that $t_2$ has to have a leaf that has a singleton $X_3$, therefore it is satisfied as well.

\indent Finally, suppose that there are three distinct source $st$-labels. In this case we use the encoding specified in Table \ref{table:encoding} on the subgraphs of each source $st$-label. It is clear on inspection that $\sumxi$ can be recovered from any two of the received values (as from any two of the linear combinations stated, one can deduce the sum $X_1+X_2+X_3$).

\begin{table}[t]
\caption{Encoding on subgraphs of different source $st$-labels. Recovery of $\sumxi$ is possible from any two of the received values, using additions or subtractions.}
\centering
\begin{tabular}{c c}
\hline\hline
Source $st$-label & Encoding\\
\hline
$(s_1, s_2)$ & $2X_1 + X_2$\\
$(s_2, s_3)$ & $X_2 + 2X_3$\\
$(s_1, s_3)$ & $X_1 - X_3$\\
\hline
\end{tabular}
\label{table:encoding}
\end{table}

\item[c)] The degree sequence is a permutation of $(1, 2, 3)$.\\
In this case (see Figure~\ref{fig:aux_bipartite_graph_2}), the degree sequence dictates that there have to be two terminals that share two $st$-labels (namely, two terminals that together appear in two
different $st$-labels). This implies that the source label of those $st$-labels has to be different. For the subgraphs induced by these $st$-labels, we use the encoding proposed in Table \ref{table:encoding}. For the subgraph induced by the remaining $st$-label, we perform greedy encoding. For example, suppose that the $st$-labels are $(s_1, s_2, t_1, t_2)$, $(s_2, s_3, t_1, t_2)$ and $(s_2, s_3, t_1, t_3)$. As shown in Figure \ref{fig:aux_bipartite_graph_2}, $t_1$ and $t_2$ are clearly satisfied (even without using the information from $st$-label $(s_2, s_3, t_1, t_3)$). Terminal $t_3$ has to have a singleton leaf containing $X_1$ by the connectivity condition and is therefore satisfied.
\end{itemize}
Together, these arguments establish that in the case when there are three $st$-labels, all terminals can be satisfied.
\item[(ii)] {\it Case 3.} There exist more than three distinct $st$-labels in $G$.\\
Note that if there are at least four $st$-labels in $G$, then (a) there are two $st$-labels with the same terminal label, since there are exactly three possible terminal labels, and (b) for the $st$-labels with the same terminal labels, the source labels necessarily have to be different. Our strategy is as follows. For the terminals that share two $st$-labels, use the encoding proposed in Table \ref{table:encoding}. If the remaining terminal has access to only one source $st$-label, then use greedy encoding and note that this terminal has to have a singleton leaf. If it has access to at least two source $st$-labels, simply use the encoding in Table \ref{table:encoding} for it as well.
\end{itemize}
\end{proof}
It remains to develop the argument in the case when there are exactly two distinct $st$-labels in $G$. For this we need to explicitly use the fact that
there are two edge-disjoint paths between each $s_i - t_j$ pair.

\begin{lemma}
\label{lem:final_case}
Consider a graph $G$, with sources, $s_i, i = 1, \dots, 3$, and terminals $t_j, j = 1, \dots 3$, such that (a) it does not have any $(3,3), (2,3)$ or $(3,2)$ nodes, and (b) there exist at least two $s_i - t_j$ paths for all $i$ and $j$. Consider the set of all $(2,2)$ nodes in $G$ and their corresponding $st$-labels. If there exist exactly two distinct $st$-labels in $G$, then there exists a set of coding vectors such that each terminal can recover $\sumxi$.
\end{lemma}
\begin{proof}
As in the proof of Lemma \ref{lemma:one_or_three_colors}, we argue based on the content of the leaves of the terminals. Suppose that the auxiliary bipartite graph $G_{aux}$ is formed. If both the $st$-labels have the same terminal label (see Figure \ref{fig:last_case_1} for an example), then it is clear that the encoding in Table \ref{table:encoding} on the subgraphs induced by the $st$-labels suffices for the corresponding terminals. The third terminal has singleton leaves corresponding to each source and can compute $\sumxi$.

\begin{figure}[t]
\centering
\includegraphics[width=80mm, clip=true]{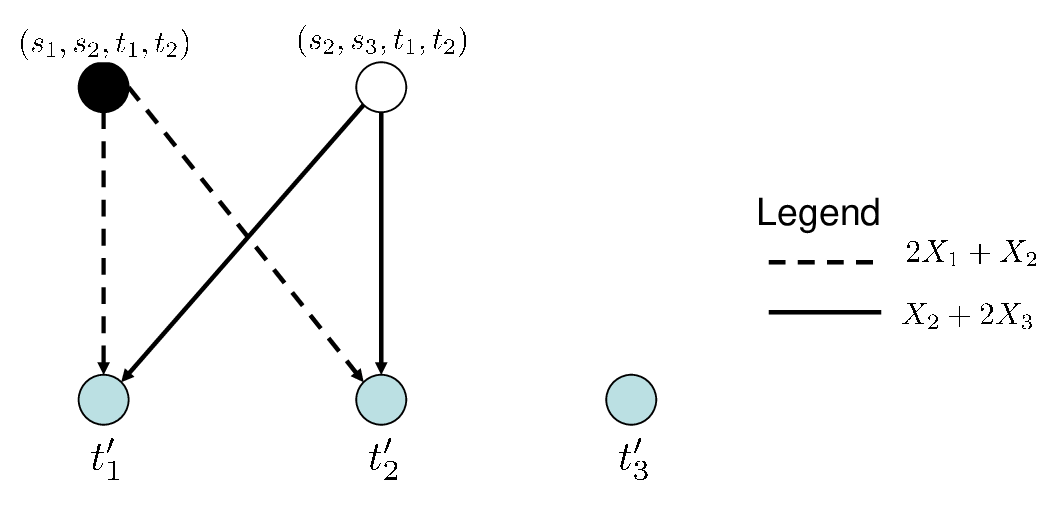} \vspace{-3mm}
\caption{\label{fig:last_case_1}An instance of $G_{aux}$ when there exist exactly two distinct $st$-labels under case 3, such that the terminal labels of the $st$-labels are the same.}
\end{figure}

Another possibility is that the terminal labels of the $st$-labels are different, but the source labels are the same. It should be clear that this case can be handled by greedy encoding on the $st$-labels.

The situation is more complicated when the terminal and source labels of the $st$-labels are different, see for example Figure \ref{fig:last_case_2}. In the case depicted, greedy encoding does not work since it satisfies $t_1$ and $t_3$ but not $t_2$. W.l.o.g., we assume that the $st$-labels are $(s_1, s_2, t_1, t_2)$ and $(s_2, s_3, t_2, t_3)$.
Now, we know that there exist two vertex-disjoint paths between $s_1$ (a similar argument can be made for $s_3$) and $t_2$. Each of these paths has a leaf for $t_2$. If one of the leaves is a $(1,\cdot)$ leaf that contains a singleton $X_1$, 
then performing greedy encoding on the two $st$-labels works since $t_2$ obtains $X_1 + X_2$, $X_1$ and $X_2 + X_3$ and the other terminals will obtain singleton leaves that satisfy their demand. Likewise, if there is a singleton leaf containing $X_3$ on the vertex disjoint paths from $s_3$ to $t_2$, then greedy encoding works.


Thus, the corresponding leaves of $t_2$ must be of type $(2,2)$.
This implies that there are at least four distinct leaves of $t_2$ of type $(2,2)$, two of $st$-label $(s_1, s_2, t_1, t_2)$ and two of $st$-label $(s_2, s_3, t_2, t_3)$. 
Our proof is concluded by the following claims.

Consider the subgraph induced by nodes labeled by one of the $st$-labels above, w.l.o.g. $(s_1, s_2, t_1, t_2)$, in $G$ together with the $(1,\cdot)$ nodes connected to either $s_1$ or $s_2$ in $G$. Denote this subgraph by $G'$.  Consider a random linear network code on the nodes of  $G'$ (namely, each node outputs a random linear combination of its incoming information over the underlying finite field of size $q$). Let $q = 2^m$. We show, with high probability (given $m$ is large enough), that such a code allows both $t_1$ and $t_2$ to receive two linearly independent combinations of $X_1$ and $X_2$ at their leaves. An analogous argument also holds for $t_2$ and $t_3$ when considering the $st$-label $(s_2, s_3, t_2, t_3)$ and the information $X_2$ and $X_3$. This suffices to conclude our assertion. Our proof is based on the following two claims.
\begin{claim}
\label{claim:random1}
Let $u$ be any leaf in $G'$. Let $U=\alpha X_1+\beta X_2$ be the incoming information of $u$. With probability $(1-2^{-m+1})^{|V|}$ both $\alpha$ and $\beta$ are not zero.
\end{claim}

\begin{proof}
Denote by $C = \{c_i\}$ the multiset of coefficients used in the random linear network code on $G'$. Namely, each $c_i$ is uniformly distributed in $GF(2^m)$, and the information on each edge $e$ is a linear combination of it's incoming information using coefficients from $C$ (each coefficient in $C$ is used only once).

It is not hard to verify that $\alpha$ is a multivariate polynomial in the variables in $C$ of total degree $\ell$, where $\ell$ is the length of the longest path between $s_i$ and $u$ (here $i=1,2$). Namely, $\ell \leq n = |V|$. Moreover, each variable $c_i$ in $\alpha$ is of degree at most $1$.
As $u$ is a $(2,2)$ leaf and is connected to $s_1$, there is a setting for the variables in $C$ such that $\alpha \ne 0$ (consider for example setting the values of variables in $C$ to match the greedy encoding function discussed previously). Thus, $\alpha$ is not the zero polynomial.
We conclude, using Lemma 4 of \cite{hoMKKESL06}, that $\alpha$ obtains that value $0$ with probability at most $1-(1-2^{-m})^n$ (over the choice of the values of variables in $C$).
(We note that Lemma 4 of \cite{hoMKKESL06} is a slightly refined version of the Schwartz-Zippel lemma.)
The same analysis holds for $\beta$.
Finally, to study the probability that either $\alpha$ or $\beta$ are zero we study the polynomial $\alpha \cdot \beta$, of total degree $2\ell$, where each variable $c_i$ in $\alpha\cdot\beta$ is of degree at most $2$.
Our assertion now follows from Lemma 4 of \cite{hoMKKESL06}.
\end{proof}

\begin{figure}[t]
\centering
\includegraphics[width=50mm, clip=true]{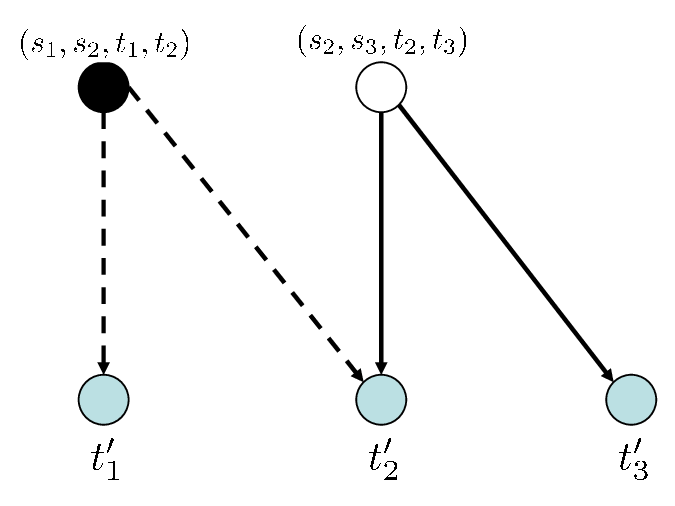} \vspace{-3mm}
\caption{\label{fig:last_case_2} An instance of $G_{aux}$ when there exist exactly two distinct $st$-labels under case 3, such that both the source labels and the terminal labels of the $st$-labels are different.}
\end{figure}
\begin{claim}
\label{claim:random2}
Consider the terminal $t_2$ and its two edge disjoint paths from $s_1$ denoted $P_1$ and $P_2$.
Let $u_1$ and $u_2$ be the corresponding leaves on paths $P_1$ and $P_2$.
Let $U_1=\alpha_1 X_1+\beta_1 X_2$ be the incoming information of $u_1$, and $U_2=\alpha_2 X_1+\beta_2 X_2$ the incoming information of $u_2$. With probability $(1-2^{-m+1})^n$ the vectors $\{(\alpha_i,\beta_i)\}_{i=1,2}$ are independent.
\end{claim}

\begin{proof}
We first note that, as the leaves of $t_2$ are of type $(2,2)$ and as both $u_1$ and $u_2$ are connected to $s_1$ it holds that both $u_1$ and $u_2$ are of $st$-label $(s_1, s_2, t_1, t_2)$ and in $G'$.
Our proof now follows the line of proof given in Claim~\ref{claim:random1}. Namely, let $C = \{c_i\}$ be the multiset of coefficients used in the random linear network code on $G'$. As before, $\alpha_1$, $\alpha_2$, $\beta_1$ and $\beta_2$ are multivariate polynomials in the variables in $C$. To study the independence between $U_1$ and $U_2$ we study the determinant $\Gamma$ of the $2\times 2$ matrix with rows $(\alpha_1,\beta_1)$, and $(\alpha_2,\beta_2)$.
The determinant $\Gamma$ is of total degree $2\ell$, where each variable $c_i$ in $\Gamma$ is of degree at most $2$.
So to conclude our assertion (via Lemma 4 of \cite{hoMKKESL06}) it suffices to prove that $\Gamma$ is not the zero polynomial.

To this end, we present an encoding function (a setting of assignments for the variables in $C$) for which $\Gamma$ will be $1$.
Consider the two disjoint paths connecting $s_1$ and terminal $t_2$ (denoted as $P_1$ and $P_2$).
Recall that $u_1$ and $u_2$ are the corresponding leaves of $st$-label $(s_1, s_2, t_1, t_2)$, where $u_i \in P_i$.
Let $v$ be the vertex closest to $s_1$ on these paths that is connected to $s_2$ (ties broken arbitrarily), assume w.l.o.g. that $v \in P_2$.
Let $P_3$ be the path connecting $s_2$ and $v$.
Consider the subgraph $H$ of $G'$ consisting of the paths $P_1,P_2$ and $P_3$.
Using the edges of $H$ alone, one can design a routing scheme such that $u_1$ will receive the information $X_1$ 
and $u_2$ the information $X_2$. This will imply that $(\alpha_1,\beta_1)=(1,0)$, $(\alpha_2,\beta_2)=(0,1)$, and $\Gamma =1$.
Indeed, just forward $X_1$ on $P_1$ and forward $X_2$ on $P_3$ until it reaches $v$ and then from $v$ to $u_2$ on $P_2$.
\end{proof}

We are now ready to complete the proof of Lemma~\ref{lem:final_case} for the case that $t_2$ has 4 type $(2,2)$ leaves.
Namely, we show that in this case Claims~\ref{claim:random1} and \ref{claim:random2} allow sum communication when random linear network coding is applied over the network.
We start with terminal $t_2$. By Claim~\ref{claim:random2}, with high probability $t_2$ will obtain two linearly independent linear combination of $X_1$ and $X_2$ on two of its leaves and
two linearly independent linear combination of $X_2$ and $X_3$ on the other pair of leaves. This will now allow $t_2$ to obtain the summation $X_1+X_2+X_3$ by an appropriate encoding over the reversed tree of $t_2$-edges in $G$.

Next, we address terminal $t_1$.
Consider its two edge disjoint paths from $s_1$ denoted $P_1$ and $P_2$.
Let $u_1$ and $u_2$ be the corresponding leaves on paths $P_1$ and $P_2$ (to simplify notation we use the same notation as previously used for $t_2$).
Here, we consider two cases, if both $u_1$ and $u_2$ are $(2,2)$ nodes, then by Claim~\ref{claim:random2} we are done (with high probability), as in the analysis of terminal $t_2$ above.
Namely, with high probability (given $m$ large enough) $t_1$ will receive two linearly independent combinations of $X_1$ and $X_2$ at $u_1$ and $u_2$.
Otherwise, $t_1$ has at least one singleton leaf with $X_1$ exclusively. 
Denote this leaf as $v_1$.
Notice that $t_1$ must have at least a single $(2,2)$ leaf (by Claim~\ref{claim:leaf_existence}), denote this leaf by $v_2$.
Finally, by Claim~\ref{claim:random1} it holds that with high probability the information present at $v_1$ and at $v_2$ is independent.

To conclude, notice that the discussion above (when applied symmetrically for $t_3$ and the $st$-label $(s_2, s_3, t_2, t_3)$) implies that all terminals are able to obtain the desired sum $X_1+X_2+X_3$ (by an appropriate setting of the encoding functions on their $(\cdot,1)$ edges).

\end{proof} 

\section{Discussion and Future Work}
\label{sec:conclusion}
In this work, we have introduced the problem of multicasting the sum of sources over a network. 
We have shown that in networks with unit capacity edges, and unit-entropy sources, with at most two sources or two terminals, the sum can be recovered at the terminals, as long as there exists a path between each source-terminal pair. Furthermore, we demonstrate that this characterization does not hold for three sources (3s)/three terminal (3t) networks. For the $3s/3t$ case we show that if each source terminal pair is connected by at least two edge disjoint paths, sum recovery is possible at the terminals. In each of these cases we present efficient network code assignment algorithms. 


\begin{figure}[t]
\centering
\includegraphics[width=90mm, clip=true]{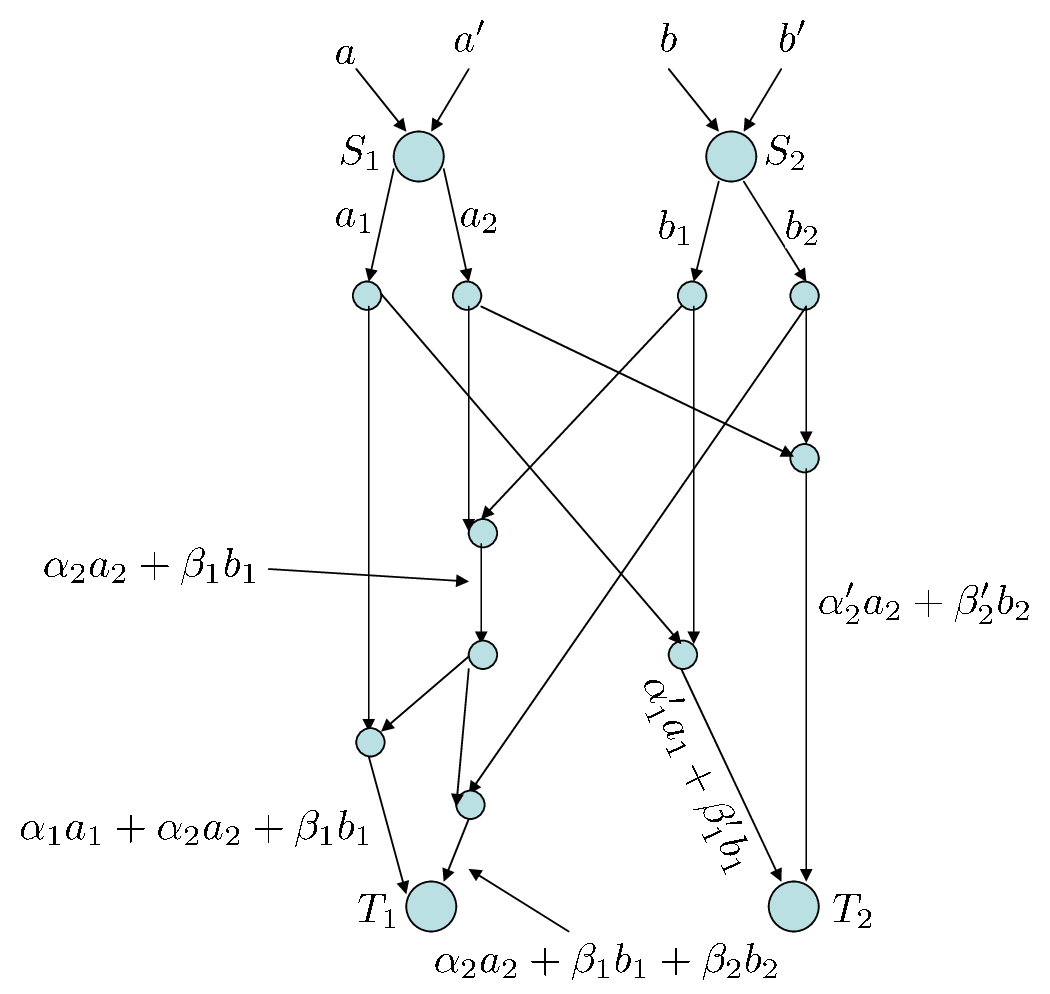} \vspace{-3mm}
\caption{\label{fig:higher-entropy-counter-eg}
Example of a network with two sources and two terminals, such that there exist two edge-disjoint paths between each source and each terminal. Source node $S_1$ ($S_2$) observes a source of entropy 2, $[a ~ a']$ ($[b ~ b']$). The terminals seek to reconstruct $[a + b ~ a' + b']$. However, this is impossible with linear codes.}
\end{figure}

Several questions remain open, that we discuss below.
\begin{itemize}
\item As our techniques do not seem to extend to the case of a higher number of sources and terminals, at present, the case of $|S| > 3$ and $|T| > 3$ is completely open. 
\item In our problem formulation, we have considered unit-entropy sources over unit-capacity networks. However, in general, one could consider sources of arbitrary entropies, by considering vector-sources (as considered in \cite{appuswamyFKZ08}), and requiring the terminals to recover a vector that contains component-wise function evaluations. This version of the problem is also open for the most part. In fact, in this case even our characterization for $|S| = 2$ does not hold. For example, consider the two-sources, two-terminals network shown in Figure \ref{fig:higher-entropy-counter-eg}, where each edge is of unit-entropy. Each source node observes a source of entropy two, that is denoted by a vector of length two. The terminals need to recover the vector sum.

In this network there are two sources, and based on our result in Section \ref{sec:case_2s_or_2t} it is natural to conjecture that if max-flow $(s_i - t_j) = 2$, holds for $i,j = 1,2$, then a network coding assignment exists. The network in Figure \ref{fig:higher-entropy-counter-eg} has this connectivity requirement. However, as shown in the Appendix, using linear codes to recover the vector sum at both the terminals is not possible.

\item We have exclusively considered the case of directed acyclic networks. An interesting direction to pursue would be to examine whether these characterizations hold in the case of networks where directed cycles are allowed.
\item Our work has been in the context of zero-error recovery of the sum of the sources. It would be interesting to examine whether the conclusion changes significantly if one allows for recovery with some (small) probability of error.
\end{itemize}

\appendix

\subsection{Discussion about network in Figure \ref{fig:higher-entropy-counter-eg}}
We prove that under linear network coding, recovering $[a+b ~ a'+b']$ at $T_1$ and $T_2$ is impossible.
We use the notation of Figure \ref{fig:higher-entropy-counter-eg}.
Let $A_1$ and $B_1$ be matrices such that $\begin{bmatrix} a_1\\ a_2 \end{bmatrix} = A_1 \begin{bmatrix} a \\ a' \end{bmatrix}$ and $\begin{bmatrix} b_1\\ b_2 \end{bmatrix} = B_1 \begin{bmatrix} b \\ b' \end{bmatrix}$. Without loss of generality, we can express the received vectors at terminals $T_1$ and $T_2$ as
\begin{align*}
z_{T_1} &= \begin{bmatrix} \alpha_1 & \alpha_2 & \beta_1 & 0\\
0 & \alpha_2 & \beta_1 & \beta_2
\end{bmatrix} \begin{bmatrix} a_1 \\ a_2 \\ b_1 \\ b_2\end{bmatrix}\\
z_{T_2} &= \begin{bmatrix} \alpha_1' & 0 & \beta_1' & 0\\
0 & \alpha_2' & 0 & \beta_2'
\end{bmatrix} \begin{bmatrix} a_1 \\ a_2 \\ b_1 \\ b_2\end{bmatrix}
\end{align*}
Using simple computations it is not hard to see that for both the terminals to be able to recover $[a+b ~a'+b']^T$ we need
\begin{align*}
\begin{bmatrix} \alpha_1 & \alpha_2 \\ 0 & \alpha_2 \end{bmatrix} A_1 &= \begin{bmatrix} \beta_1 & 0\\ \beta_1 & \beta_2 \end{bmatrix} B_1, \textrm{~and}\\
\begin{bmatrix} \alpha_1' & 0 \\ 0 & \alpha_2' \end{bmatrix} A_1 &= \begin{bmatrix} \beta_1' & 0\\ 0 & \beta_2' \end{bmatrix} B_1, \textrm{~and}
\end{align*}
require all these matrices to be full-rank. Note that the full-rank condition requires that all the coefficients $\alpha_1, \alpha_2, \beta_1, \beta_2, \alpha_1', \alpha_2', \beta_1'$ and $\beta_2'$ be non-zero and the matrices $A_1$ and $B_1$ to be full-rank. In particular, the required condition is equivalent to requiring that
\begin{align*}
{\begin{bmatrix} \alpha_1 & \alpha_2 \\ 0 & \alpha_2 \end{bmatrix}}^{-1} \begin{bmatrix} \beta_1 & 0\\ \beta_1 & \beta_2 \end{bmatrix} &= {\begin{bmatrix} \alpha_1' & 0 \\ 0 & \alpha_2' \end{bmatrix}}^{-1} \begin{bmatrix} \beta_1' & 0\\ 0 & \beta_2'\end{bmatrix}\\
\Rightarrow \frac{1}{\alpha_1 \alpha_2} \begin{bmatrix} \alpha_2 & -\alpha_2 \\ 0 & \alpha_1 \end{bmatrix} \begin{bmatrix} \beta_1 & 0\\ \beta_1 & \beta_2 \end{bmatrix} &= \begin{bmatrix} 1/\alpha_1' & 0 \\ 0 & 1/\alpha_2' \end{bmatrix} \begin{bmatrix} \beta_1' & 0\\ 0 & \beta_2'\end{bmatrix}\\
\Rightarrow \begin{bmatrix} 0 & -\frac{\beta_2}{\alpha_1} \\ \frac{\beta_1}{\alpha_2} & \frac{\beta_2}{\alpha_2}\end{bmatrix} &= \begin{bmatrix} \frac{\beta_1'}{\alpha_1'} & 0 \\ 0 & \frac{\beta_2'}{\alpha_2'}\end{bmatrix}
\end{align*}
For the above equality to hold, we definitely need $\beta_2 = 0$, but this would contradict the requirement that $\beta_2 \neq 0$ that is needed for the full rank condition.
\endproof

\begin{biography}[{\includegraphics[width=1in,
height=1.25in,clip,keepaspectratio]{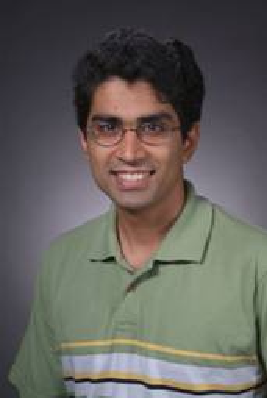}}]{Aditya
Ramamoorthy} (M'05) received the B.Tech. degree in electrical engineering
from the Indian Institute of Technology, Delhi, in 1999, and the M.S.
and Ph.D. degrees from the University of California, Los Angeles (UCLA), in
2002 and 2005, respectively.
He was a systems engineer with Biomorphic VLSI Inc. until 2001. From 2005
to 2006, he was with the Data Storage Signal Processing Group of Marvell
Semiconductor Inc. Since fall 2006, he has been with the Electrical and Computer Engineering Department at Iowa State University,
Ames, IA 50011, USA. His research interests are in the areas of network information theory,
channel coding and signal processing for storage devices, and its applications
to nanotechnology.

Dr. Ramamoorthy is the recipient of the 2012 Iowa State University's Early Career Engineering Faculty Research Award, the 2012 NSF CAREER award, and the Harpole-Pentair professorship in 2009 and 2010. He has been serving as an associate editor for the IEEE Transactions on Communications since 2011.
\end{biography}

\begin{biography}{Michael Langberg} (M'07) is an Associate Professor in the Mathematics and Computer Science department at the Open University of Israel. Previously, between 2003 and 2006, he was a postdoctoral scholar in the Computer Science and Electrical Engineering departments at the California Institute of Technology. He received his B.Sc. in mathematics and computer science from Tel-Aviv University in 1996, and his M.Sc. and Ph.D. in computer science from the Weizmann Institute of Science in 1998 and 2003 respectively. His research interests include information theory, combinatorics, and algorithm design.
\end{biography}

\end{document}